\documentclass[12pt]{amsart}
\usepackage{hyperref, amsmath,amssymb,amsfonts,amsbsy,mathtools,cite,setspace}
\usepackage{graphicx}
\usepackage{amssymb}
\usepackage{amsthm}
\usepackage{enumerate}
\usepackage{wasysym}
\usepackage{cite}
\usepackage{tikz}
\usepackage{mathtools}

\usepackage{braket}
\usepackage{color}

\newcommand{\la}{\langle}
\newcommand{\ra}{\rangle}

\def\ket#1{| #1 \rangle}
\def\bra#1{\langle #1 |}
\def\kb#1#2{|#1\rangle\!\langle #2 |}

\newtheorem{defn}{Definition}
\newtheorem{prop}{Proposition}
\newtheorem{thm}{Theorem}

\newtheorem{lem}{Lemma}

\newtheorem{rmk}{Remark}
\newtheorem{cor}{Corollary}

\newcommand{\C}{{\mathbb C}}

\newcommand{\M}{{\mathbb M}}
\renewcommand{\H}{{\mathcal H}}
\renewcommand{\S}{{\mathcal S}}
\newcommand{\operp}{$\bigcirc$\kern-.91em{$\perp$}}

\newcommand{\vecop}{\operatorname{vec}}
\newcommand{\tr}{\operatorname{Tr}}
\newcommand{\spn}{\operatorname{span}}

\def\be{\begin{eqnarray}}
\def\ee{\end{eqnarray}}
\def\bee{\begin{eqnarray*}}
\def\eee{\end{eqnarray*}}

\begin{document}
\title[Operator Structures and One-Way LOCC Conditions]{Operator Structures and Quantum One-Way LOCC Conditions}
\author[D.W.Kribs, C.Mintah, M.Nathanson, R.Pereira]{David W. Kribs$^{1,2}$, Comfort Mintah$^{1}$, Michael Nathanson$^3$, Rajesh Pereira$^{1}$}

\address{$^1$Department of Mathematics \& Statistics, University of Guelph, Guelph, ON, Canada N1G 2W1}
\address{$^2$Institute for Quantum Computing and Department of Physics \& Astronomy, University of Waterloo, Waterloo, ON, Canada N2L 3G1}
\address{$^3$Department of Mathematics and Computer Science, Saint Mary's College of California, Moraga, CA, USA 94556}

\begin{abstract}
We conduct the first detailed analysis in quantum information of recently derived operator relations from the study of quantum one-way local operations and classical communications (LOCC). We show how operator structures such as operator systems, operator algebras, and Hilbert $C^*$-modules all naturally arise in this setting, and we make use of these structures to derive new results and new derivations of some established results in the study of LOCC. We also show that perfect distinguishability under one-way LOCC and under arbitrary operations is equivalent for several families of operators that appear jointly in matrix and operator theory and quantum information theory.
\end{abstract}

\subjclass[2010]{47L90, 46B28, 81P15, 81P45, 81R15}

\keywords{quantum states, local operations and classical communication, perfect distinguishability, operator system, operator algebra, separating vector, Hilbert module.}

\date{December 8, 2016}

\maketitle

\section{Introduction}
In quantum information theory, the paradigm of Local Operations and Classical Communications (LOCC) is of paramount interest for both theory and application. It is a way  of probing the interplay between locality and quantum entanglement, and its operational definition demonstrates the possibilities and limitations of distributed quantum algorithms. By LOCC, we indicate a situation in which a number of parties share an initial quantum state on which each party can perform local quantum operations. The parties are allowed to coordinate their measurements and compare their results by transmitting classical information. The present work will only concern itself with LOCC between {\it two} parties.

The set of LOCC operations on a bipartite system is notoriously difficult to characterize mathematically, and we sometimes restrict ourselves to the case of one-way LOCC, in which the communication is limited to one predetermined direction. This subset of LOCC operations includes many well-known protocols, such as quantum teleportation \cite{Teleportation}, but there are known examples of tasks which require two-way communication in order to be accomplished with LOCC\cite{GL03, bennett1999quantum, n13, OwariHayashi}.

A particularly instructive problem for LOCC operations is that of distinguishing pure quantum states. Given an unknown representative $\ket{\phi}$ from a  set of states $\S$,  $\ket{\phi}$ can always be identified using quantum operations if and only if the elements of $\S$ are mutually orthogonal. No such simple characterization has been found for general LOCC, although many results have been obtained for specific cases (such as \cite{Walgate-2002, Entanglement-horodecki, Ghosh-2001, Nathanson-2005}). In the case of one-way LOCC, a mathematical characterization in terms of operator relations was recently obtained in \cite{n13}. In this paper, we conduct the first detailed analysis of these relations in the context of quantum information. We show how operator structures such as operator systems, operator algebras, and Hilbert $C^*$-modules all naturally arise in this setting, and we make use of these structures to derive new results and new derivations of some established results in the study of LOCC. We also show that perfect distinguishability under one-way LOCC and under arbitrary operations are equivalent for several families of operators that appear jointly in matrix and operator theory and in quantum information theory.

The paper is organized as follows. In Section~\ref{S:operatorrelations} we give the mathematical description of LOCC in terms of quantum operations and derive the operator conditions for one-way LOCC. We follow this in Section~\ref{S:perfect} by showing that perfect distinguishability under one-way LOCC is equivalent to the same under arbitrary operations for selected families of operators. In Section~\ref{S:operatorseparating} we draw a connection between one-way LOCC and the study of operator algebras and separating vectors, based around an analysis of operator systems that naturally arise in the study. We conclude in Section~\ref{S:hilbertmodule} by considering in more detail a distinguished special case of one-way LOCC that generates Hilbert C$^*$-module structures.

\section{LOCC Description via Operator Relations}\label{S:operatorrelations}

A basic scenario in quantum communication occurs when two parties, Alice and Bob say, each in control of quantum system, represented on finite-dimensional Hilbert spaces $\H_A$ and $\H_B$. They are physically separated  so that their measurement protocols are restricted to those only using local quantum operations and classical communications (LOCC). Their joint system $\H_A \otimes \H_B$ has been prepared in a pure state from a known set of (possibly entangled) states $\S = \{ \ket{\psi_i}  \}$, and Alice and Bob would like to determine the value of $i$ only using LOCC.

We will focus on the problem of distinguishing these states perfectly. We will also specify our analysis to the class of ``one-way'' LOCC, taking our main motivation from the work \cite{n13}. In one-way LOCC, Bob can adapt his measurement based on classical information sent to him by Alice. As a simple example, consider the pair of Bell states;
\[
\ket{\psi_1} = \frac{1}{\sqrt{2}} \big( \ket{0}_A\ket{0}_B + \ket{1}_A\ket{1}_B \big)  \quad \ket{\psi_2} = \frac{1}{\sqrt{2}} \big( \ket{0}_A\ket{1}_B + \ket{1}_A\ket{0}_B \big).
\]
Suppose that Alice measures the first qubit and obtains the result 0, and then sends the result to Bob over a classical channel. Bob then measures the second qubit, and suppose he also obtains 0. The joint measurement outcome is thus $\ket{0}_A\ket{0}_B$, and hence Bob knows the state given to them was $\ket{\psi_1}$. In this example the protocol is especially simple, since Bob's measurement does not depend on Alice's outcome. In general, we will allow Bob to adapt his measurements depending on Alice's outcome.

Mathematically then, a one-way LOCC measurement is of the form $\M = \{ A_k \otimes B_{k,j} \}$, with the positive operators making up the measurement outcomes satisfying $\sum_k A_k = I_A$ and $\sum_{j} B_{k,j} = I_B$ for each $k$. If the outcome $A_k \otimes B_{k,j}$ is obtained for any $k$, the conclusion is that the prepared state was $\ket{\psi_j}$. In the simple example above, we have $A_1 = \kb00$, $A_2 = \kb11$, and $B_{i,j} = A_{2-\delta_{ij}}$ (where $\delta_{ij}$ is the Kronecker delta). Note that this rule will perfectly distinguish the states if for all $k$ and $i\neq j$, 
\[
\bra{\psi_i} A_k \otimes B_{k,j} \ket{\psi_i} = 0.
\]
If $\ket{a\otimes b}$ is an eigenvector of $A_k \otimes B_{k,j}$ with non-zero eigenvalue, this implies $|\la\psi_i | a \otimes b\ra|^2 =0$. Hence, without loss of generality, we can assume each $A_k$ is a scalar multiple of a rank one projection.

We are now prepared to prove the matrix theoretic characterization from \cite{n13} of one-way LOCC, with algebraic relations that give the main motivation for our subsequent analysis. Notationally, we shall use $\ket{\Phi}$ to denote the standard maximally entangled state $$\ket{\Phi} = \frac{1}{\sqrt{d}}\big( \ket{00}+\ldots + \ket{d-1 \, d-1}\big),$$ on two-qudit Hilbert space $\C^d\otimes\C^d$. We also note a distinguished special case of the following hypotheses, by recalling that any two-qudit maximally entangled state can be written as $(I\otimes U)\ket{\Phi}$ where $I$ is the identity operator and $U$ is a unitary on $\C^d$.

\begin{prop}\label{prop1}
Let $\S = \{ \ket{\psi_i}= (I\otimes M_i)\ket{\Phi}  \} \subseteq \C^d\otimes\C^d$ be a set of orthogonal states. Then the following statements are equivalent:
\begin{itemize}
\item[$(i)$] The elements of $\S$ can be perfectly distinguished with one-way LOCC.
\item[$(ii)$] There exists a set of states $\{\ket{\phi_k}\}_{k=1}^r\subseteq\C^d$ and positive numbers $\{ m_k \}$ such that $\sum_{k} m_k \kb{\phi_k}{\phi_k} = I$ and for all $k$ and $i \ne j$, 
\begin{equation}\label{maineqn}
\bra{\phi_k} M^*_j M_i \ket{\phi_k} = 0.
\end{equation}

\item[$(iii)$] There is a $d \times r$ partial isometry matrix $W$ such that $WW^* = I_d$, and for all $i \ne j$, every diagonal entry of the $r \times r$ matrix $W^*M_j^*M_iW$ is equal to zero.
\end{itemize}
\end{prop}

\begin{proof}
First we recall the operator $\vecop(A)$, which takes the columns of a matrix $A$ and stacks them on top of each other to form a vector. The relation $(C\otimes A) \vecop(B) = \vecop(ABC^T)$ holds whenever the matrix product on the right makes sense. Observe that $\ket{\Phi} = \sqrt{d}^{-1}\vecop(I)$ and each $\ket{\psi_i} = (I\otimes M_i) \ket{\Phi} = \sqrt{d}^{-1} \vecop(M_i)$.

Now suppose condition $(i)$ holds. By the preceding discussion, we may assume Alice and Bob's measurements are defined by rank one operators, and hence there are states $\ket{a_k}$, $1 \leq k \leq r$, with $\tr(\kb{a_k}{a_k})=1$, such that Alice's measurement operators on the composite system are given by the set $\{A_k\otimes I\}$ with each $A_k = m_k \kb{a_k}{a_k}$. Alice's measurement must preserve the orthogonality of the unknown states, and her  measurement outcome $k$ corresponds to the state $(A_k\otimes I)\ket{\psi_i}$ being obtained, which can be rewritten using the identification
\[
\sqrt{d} (\kb{a_k}{a_k}\otimes I)\ket{\psi_i} = (\kb{a_k}{a_k}\otimes I)\vecop(M_i) = \vecop(M_i  \kb{\overline{a_k}}{\overline{a_k}} ),
\]
where we have used the notation $\ket{\overline{a_k}}$ to denote the state with coordinates given by the complex conjugates of the coordinates of $\ket{a_k}$. We can then use the relation $\vecop(\kb{a}{b}) = \tr_A( \vecop(\kb{a}{b}) \vecop(\kb{a}{b})^*)$ for all states $\ket{a},\ket{b}$ on $\C^d$, to obtain the corresponding state on Bob's system as,
\begin{eqnarray*}
\tr_A ( (A_k \otimes I ) \kb{\psi_i}{\psi_i}(A_k^*\otimes I) &=& 
 d^{-1} M_i\kb{\overline{a_k}}{\overline{a_k}} M_i^*.
\end{eqnarray*}
These states can be perfectly distinguished by Bob precisely when they are orthogonal; in other words, the algebraic relations of Eq.~(\ref{maineqn}) are satisfied with $\ket{\phi_k} = \ket{\overline{a_k}}$ for all $k$. This establishes condition~$(ii)$, and one can verify that this argument is reversible, thus establishing that $(ii)$ implies $(i)$ as well.

When condition $(ii)$ holds, we can define a partial isometry $W^*: \C^d \rightarrow \C^r$ as the sum of outer products
\[
W = \sum_{k=1}^r \sqrt{m_k} \ket{\phi_k} \bra{k-1},
\]
and one can verify $WW^* = I_d$. Moreover, whenever $i \ne j$, the $k$th diagonal entry satisfies:
\[
\bra{k} W^{*}M_{j}^{*}M_{i}W \ket{k} = m_k \bra{\phi_{k}} M_{j}^{*} M_{i} \ket{\phi_k} = 0.
\]
On the other hand, given a matrix representation of a partial isometry $W: \C^r \rightarrow \C^d$ with diagonal entries satisfying condition $(iii)$, we can use the corresponding outer product representation to define the states $\{ \ket{\phi_k} \}$ that satisfy $(ii)$.
\end{proof}

\section{Perfect Distinguishability under One-Way LOCC vs Arbitrary Operations}\label{S:perfect}

For a general class of states perfect distinguishability by one-way LOCC is a much stronger condition than perfect distinguishability using arbitrary operations.  A collection of pure states is distinguishable under arbitrary operations if and only if the states form an orthonormal set. Specifically then, it follows that a collection of states $\{ (I\otimes M_i)\ket{\Phi}\}_{i}$ is distinguishable under arbitrary operations if and only if $\tr(M_j^*M_i)=0$ whenever $i\neq j$.

In this section, we show that for certain classes of states and their associated defining operators, which jointly arise in operator and matrix theoretic as well as quantum information theoretic settings, perfect distinguishablity under one-way LOCC is equivalent to perfect distinguishabilty under arbitrary operations.

In what follows, we let $\Delta$ be the map that zeros out all the off-diagonal entries of a square matrix of a given size but leaves its diagonal entries unchanged; in other words, as a map on operators, there is a basis $\{ \ket{k} \}$ such that $\Delta(\rho) = \sum_k \kb{k}{k} \rho \kb{k}{k}$ is the von Neumann measurement map defined by the basis.  We also recall that a $d\times d$ {\it Latin square} is a matrix with entries from $\{ 1,2,...,d\}$ with the property that no two entries in any fixed row or fixed column are the same.

\begin{prop} Let $\{ P_k\}_{k=1}^n$ be a set of $d\times d$ permutation matrices and let $\mathcal{S}= \{(I\otimes P_i)\ket{\Phi}\}_{i}$.  Then the following statements are equivalent:
\begin{enumerate}
\item The states in $\mathcal{S}$ are perfectly distinguishable by one-way LOCC.

\item The states in $\mathcal{S}$ are perfectly distinguishable by arbitrary operations.

\item $\Delta(P_j^*P_i)=0$ whenever $i\neq j$.
\end{enumerate}
Moreover, if these conditions hold and $n = d$, then there exists a $d\times d$ Latin square $\mathcal{L}$ with the property that $\forall i,j,k \in \{ 1,2,...,d\} $ the $(i,j)$ entry of $P_k$ is one if and only if the $(i,j)$ entry of $\mathcal{L}$ is $k$.
\end{prop}

\begin{proof}
Evidently $(1)\implies(2)$ by virtue of the type of operations used, but we can also see this through the equations by choosing a $W$ from Proposition~\ref{prop1} and observing $\tr(P_j^*P_i) = \tr(P_j^*P_i WW^*) = \tr(W^* P_j^*P_i W) =0$ for all $i\neq j$.

To see $(2) \implies (3)$, note that if $\tr(P_j^*P_i)=0$, then $\Delta(P_j^*P_i)=0$ as well, since  $P_j^*P_i$ is a permutation matrix.

For $(3) \implies (1)$, simply take $r=d$ and $W$ to be the identity to see the conditions of Proposition~\ref{prop1} are satisfied.

Finally, suppose these equivalent conditions hold and that $n = d$. If both $P_i$ and $P_j$ had a one in their $(k,l)$ entry, then $P_j^*P_i$ would have its $(l,l)$ entry strictly positive, which would contradict $\tr(P_j^*P_i)=0$.  Thus we can let $\mathcal{L}$ be a $d \times d$ matrix whose  $(i,j)$ entry is $k$ if $P_k$ is the unique permutation matrix whose $(i,j)$ entry is one:  $\mathcal{L} = \sum_{k = 1}^d kP_k$, which is evidently a Latin square.
\end{proof}

We introduce the following concept first considered in \cite{hiroshima2004finding} in the context of searches for maximally correlated states, and which we now make use of in our current setting.

\begin{defn}
We say that a set of states $\{(I\otimes M_k)\ket{\psi_i}\}_{k=1}^n$ have a simultaneous Schmidt decomposition if there exists two unitary matrices $U$ and $V$ and $n$ complex diagonal matrices $D_k$ such that  for each $k$, $M_k=UD_kV$.
\end{defn}

These decompositions are not strictly speaking Schmidt decompositions and are instead called ``weak Schmidt decompositions'' in \cite{hua2014schmidt} because we are not imposing any requirement that the entries of the diagonal matrices $D_k$ be nonnegative. It was noted in  \cite{hiroshima2004finding} that for generalized Bell states, possessing a simultaneous Schmidt decomposition is a sufficient condition for distinguishability by LOCC.  We considerably strengthen and generalize this  result by showing that for any set of states that have a simultaneous Schmidt decomposition, distinguishability with arbitrary operations always implies distinguishability by one-way LOCC.

\begin{prop}\label{second} Let $\mathcal{S} = \{(I\otimes M_k)\ket{\psi_i}\}_{k=1}^n \in \mathbb{C}^d\otimes \mathbb{C}^d$ have a simultaneous Schmidt decomposition.  Then the following statements are equivalent:

\begin{enumerate}

\item The states in $\mathcal{S}$ are perfectly distinguishable by one-way LOCC.

\item The states in $\mathcal{S}$ are perfectly distinguishable by arbitrary operations.

\item There exists a $d \times d$  unitary matrix $A$ such that $\Delta(A^*M_j^*M_iA)=0$ whenever $i\neq j$.

\end{enumerate}
\end{prop}

\begin{proof} The implication $(1)\implies (2)$ is straightforward as above, and $(3)\implies (1)$ follows from Proposition~\ref{prop1}. So suppose $(2)$ holds, that is $\tr(M_j^*M_i)=0$ whenever $i\neq j$.  Let $U$ and $V$ be unitary matrices and $D_k$ be diagonal matrices such that $M_k=UD_kV$ for all $k$.  Then $VM_j^*M_iV^*$ is diagonal for all $i,j$. Let $F$ be the $d \times d$ Fourier matrix.  Then $F^*VM_j^*M_iV^*F$ is a trace zero circulant matrix when $i\neq j$. (More on Fourier and circulant matrices can be found in the standard reference 
\cite{davis2012circulant}). Hence condition $(3)$ is satisfied for a choice of $A=V^*F$.  (In fact, we can let $A=V^*FD$ for any diagonal unitary matrix $D$, so the solution is far from unique.)
\end{proof}

Consider the following physically motivated class of examples to which this result applies.

\begin{defn} Let $\mathcal{H}=\mathbb{C}^d$ be the qudit Hilbert space  with orthonormal basis $\{ \vert k\rangle \}_{k=0}^{d-1}$.  Let $\omega= e^\frac{2\pi i}{d}$.  We define two unitary operators $X$ and $Z$ on $\mathcal{H}$ as follows: $X\vert k\rangle =\vert k+1 (mod\ d) \rangle$ and $Z\vert k \rangle= \omega \vert k\rangle$.  Then the Generalized Pauli operators are the set $\{ X^aZ^b \}_{a,b=0}^{d-1}$. \end{defn}

Both of the sets $\{X^k\}_{k=0}^{d-1}$ and $\{Z^k\}_{k=0}^{d-1}$ form Abelian groups. It is easy to see that they satisfy the hypothesis and the second equivalent condition of Proposition \ref{second} (with these operators playing the role of the operators $M_i$); hence they satisfy all equivalent conditions of the proposition.  A more general result along these lines can be found in \cite{hiroshima2004finding}.

We can prove one more result of this type, giving an alternate proof of a known result. We first need the following result of Fillmore \cite{f69}.

\begin{lem}  Let $M$ be any trace zero matrix, then there exists a unitary matrix $V$ such that $\Delta(V^*MV)=0$.\end{lem}

We can use this lemma to prove our result for pairs of states arising from two matrices (which need not be unitary). For clarity we simply refer to the operators that define the corresponding states in matrix form.

\begin{prop}\label{third} Let $M_1$ and $M_2$ be two $d\times d$ complex matrices.  Then the following statements are equivalent:

\begin{enumerate}

\item There exists a $d \times d$  unitary matrix $V$ such that $\Delta(V^*M_2^*M_1V)=0$.

\item There exists an integer $r\ge d$, and a co-isometry $W$ from $\mathbb{C}^r$ to $\mathbb{C}^d$ such that $\Delta(W^*M_2^*M_1W)=0$.

\item $\tr(M_2^*M_1)=0$.

\end{enumerate}

\end{prop}


Thus we have recovered the remarkable result of Walgate, et al.\cite{Walgate-2000}, that {\it any} two orthogonal pure states can be distinguished by one way LOCC. We note that (3) no longer implies either (1) or (2) when the number of  general unitary matrices increases.  The example from \cite{n13} shows that this fails for a specific set of three generalized permutation matrices.

\section{Operator Algebras, Operator Systems, and Separating Vectors in the LOCC Context}\label{S:operatorseparating}

Here we exhibit a connection between perfect distinguishability with one-way LOCC and separating vectors of operator algebras and systems. Before continuing, let us recall briefly the basic structure theory for finite-dimensional $C^*$-algebras; namely, that every such algebra $\mathfrak{A}$ is $\ast$-isomorphic to the orthogonal direct sum of complex full matrix algebras $M_n$, and from the representation theory for such algebras that $\mathfrak{A}$ is unitarily equivalent to an orthogonal direct sum of the form $\bigoplus_{i}(I_{k_i}\otimes M_{n_i})$ (the indices $k_i$ correspond to multiplicities of $n_i$-dimensional irreducible representations that determine the structure of $\mathfrak{A}$). The algebra $\mathfrak{A}$ is unital if it contains the identity operator.

\begin{defn} Let $\mathfrak{A}$ be a unital $C^*$-algebra.  Any linear subspace of $\mathfrak{S}$ which contains the identity and is closed under taking adjoints is called an operator system. \end{defn}

\begin{defn} Let $\mathcal{H}$ be a Hilbert space and let $\mathfrak{S}\subseteq B(\mathcal{H})$ be a set of operators on $\mathcal{H}$ that form an operator system.  A vector $\ket{\psi}\in \mathcal{H}$ is said to be a separating vector of $\mathfrak{S}$ if $A\ket{\psi}\neq 0$ whenever $A$ is a nonzero element of $\mathfrak{S}$.    \end{defn}

If $\mathcal{H}$ is finite-dimensional and $\mathfrak{S}$ is in fact a $C^*$-subalgebra, then we may use the aforementioned representation theory for such algebras to determine the existence of a separating vector as follows. (See \cite[Chapter 3]{pereira2} for more in-depth investigations on this topic.)

\begin{prop} \label{Cex} The $C^*$-algebra $\bigoplus_{i} (I_{k_i}\otimes M_{n_i})$ has a separating vector if and only if $k_i \ge n_i$ for all $i$.\end{prop}

\begin{proof}
This can be seen most directly by first considering the special case $I_k \otimes M_n$. Suppose $k \geq n$ and hence we may choose a set of vectors $\ket{\psi_1},\ldots,\ket{\psi_k}$ that span the space $\mathbb{C}^n$. Then let $\ket{\psi} = (\ket{\psi_1},\ldots,\ket{\psi_k})^T$ and observe that $A\in M_n$ with $(I_k\otimes A)\ket{\psi}=0$ implies $A\ket{\psi_i}=0$ for all $i$, so that $A=0$ and hence $\ket{\psi}$ is a separating vector for $I_k \otimes M_n$. On the other hand, if $k < n$, then for any vector $\ket{\psi} = (\ket{\psi_1},\ldots,\ket{\psi_k})^T$ with $\ket{\psi_i}\in\mathbb{C}^n$, we can let $P$ be the (non-zero) projection onto the orthogonal complement of the subspace $\spn\{\ket{\psi_i}\}$ inside $\mathbb{C}^n$. It follows that $(I_k\otimes P)\ket{\psi} = 0$ but $P\neq 0$, and hence $I_k\otimes M_n$ has no separating vector. For the general case, a similar analysis shows that a separating vector exists if and only if $k_i \geq n_i$ for all $i$.
\end{proof}


\begin{rmk}\label{rmk: Dimension}
Suppose that $\ket{\psi}$ is a separating vector of $\mathfrak{S}$ and that $\mathfrak{S}$ is an algebra. Then a simple dimension bound may be obtained as a consequence of this result on the size of $\mathfrak{S}$ by observing that $A \to A\ket{\psi}$ is an injective linear map from $\mathfrak{S}$ to $\mathcal{H}$, and hence $dim(\mathfrak{S})\le dim(\mathcal{H})$. In particular this applies to operator systems defined by families of one-way LOCC unitaries, as noted in the corollary below.
\end{rmk}

In the context of one-way LOCC, we are interested in operator systems that arise naturally through the equations of Proposition~\ref{prop1}. We note that from the discussion above, any $C^*$-subalgebra of such an operator system must have a separating vector which limits its dimension.

\begin{thm}  \label{thm: Delta} Let $W : \mathbb{C}^r \rightarrow \mathbb{C}^d$ be an operator and let $\Delta$ be the diagonal map on $\mathbb{C}^r$ in a fixed basis as defined above. Consider the operator system $\mathfrak{S}$ on $\mathbb{C}^d$ defined as the set of all operators $X$ which satisfy:
\begin{equation}\label{sepveceqn}
\Delta(W^*XW)=\frac{\tr(X)}{d}\Delta(W^*W).
\end{equation}
If $\mathfrak{A}$ is a $C^*$-subalgebra of $\mathfrak{S}$, then $\mathfrak{A}$ has a separating vector.
\end{thm}

\begin{proof}
There exists a $k: 1\le k\le r$ such that the $(k,k)$ entry of $\Delta(W^*W)$ is a strictly positive real number, call this number $c$.  Let $\ket{\psi}$ be the $k$th column of $W$, then $\Vert A\ket{\psi}\Vert^2= \bra{k}W^*A^*AW\ket{k} = \frac{c}{d}\tr(A^*A)\neq 0$ when $A$ (and hence $A^*A$) is a nonzero element of $\mathfrak{A}$.  It follows that $\ket{\psi}$ is a separating vector for $\mathfrak{A}$.\end{proof}

In the context of LOCC, it is well-known that no more than $d$ maximally-entangled states in $\C^d\otimes\C^d$ can be distinguished with LOCC \cite{ghosh2004distinguishability, Nathanson-2005}. In the case of such a maximal set, we can state a corollary to the above result:
\begin{cor}
Let $\S = \{ \ket{\psi_i}= (I\otimes U_i)\ket{\Phi}  \}_{i = 1}^d \subseteq \C^d\otimes\C^d$ be a set of $d$ orthogonal maximally entangled states which are perfectly distinguishable with one-way LOCC, and let $\mathfrak{S}$ be the operator system spanned by the set $\{ U_j^* U_i \}$. Then $\mathfrak{S}$ has a separating vector if and only if $\dim \mathfrak{S} = d$, and in this case $\mathfrak{S}$ forms a C$^*$-algebra.
\end{cor}

\begin{proof} Since the $\ket{\psi_i}$ are mutually orthogonal, the matrices $\{ U_i \}$ are linearly independent and hence  $\dim \mathfrak{S} \ge d$. If $\dim \mathfrak{S} > d$, then $\mathfrak{S}$ does not have a separating vector by Remark \ref{rmk: Dimension}.

On the other hand, if we assume that $\dim \mathfrak{S} = d$, then $\mathfrak{S} = \mbox{span}{\{U_k\}}$ and $U_j^*U_i \in \mbox{span}{\{U_k\}}$ for all $i,j$. By the Cayley-Hamilton theorem, we can find complex polynomials $p_i(z)$ such that $U_i = p_i(U_i^*)$ for all $i$. Using these facts and the invertibility of each $U_i$, it follows that $U_i U_j = p_i(U_i^*)U_j \in \mathfrak{S}$ for each pair $i,j$, and hence that $\mathfrak{S}$ is a C$^*$-algebra.
%
Since the  $\ket{\psi_i}$ are distinguishable with one-way LOCC, there exists a $W$ such that (\ref{sepveceqn}) holds for all $X \in \{U_j^*U_i\}$ and thus for all $X \in \mathfrak{S}$. Hence $\mathfrak{S}$ has a separating vector by Theorem \ref{thm: Delta}.
\end{proof}


Note that there are cases in which a set of $d$ maximally entangled states can be distinguished with one-way LOCC but $\dim \mathfrak{S} > d$. For instance, using the generalized Pauli matrices, if we look at the set $\S = \{ U_i\}_{i = 1}^d$ with $U_i = X^i$ for $1 \le i \le d-1$ and  $U_d = Z$.  If $\ket{\phi}$ is a standard basis vector, then  $\bra{\phi} X^iZ\ket{\phi} = \bra{\phi} X^i\ket{\phi}  = 0$ if $i \ne 0$, implying that these states are LOCC-distinguishable. However, if $d>2$ then $\mathfrak{S} = \mbox{span}{\{U_j^*U_i\}}$ has dimension $3d-2$, implying that it does not have a separating vector.

We can also state a partial converse to Theorem \ref{thm: Delta}. A set of codewords $\{ X_i \}$ is {\it unambiguously distinguishable} if there exists a protocol with $(n+1)$ outcomes $\{ Y_i\}$ such that for each $i \le n$, the outcome $Y_i$ occurs with positive probability and implies that $X_i$ was sent. The outcome $Y_{n+1}$ is the error outcome and provides no conclusive information about the identity of $X_i$. A set of quantum states $\{ \ket{\psi_i} \}$ is unambiguously distinguishable if and only if they are linearly independent \cite{Chefles-1998}. If $\S = \{ \ket{\psi_i}= (I\otimes M_i\ket{\Phi}  \}$, a sufficient condition for unambiguous discrimination of $\S$ with one-way LOCC is the existence of a vector such that the vectors $\{M_i \ket{\phi} \}$ are linearly independent. This gives us the following:

\begin{thm}  Let $\mathfrak{S}$ be an operator system on $\mathbb{C}^d$ spanned by the pairwise products $\{ M_j^*M_i \}$, $i,j \in [n]$. If $\mathfrak{S}$ has a separating vector, then the bipartite states $\{ \ket{\psi_i}= (I\otimes M_i)\ket{\Phi}  \} \subseteq \C^d\otimes\C^d$ can be unambiguously distinguished using one-way LOCC.
\end{thm}
\begin{proof} If $\mathfrak{S}$ has a separating vector $\ket{\phi}$, then for any $j$, $\{ M_j^*M_i\ket{\phi} \}_{i = 1}^n$ is linearly independent which means that $\{ M_i\ket{\phi} \}$ are linearly independent. If Alice performs a measurement and gets the outcome $\ket{\overline{\phi} }\bra{\overline{\phi} }$, then Bob's system will be in the state $M_i\ket{\phi}$ for some $i$. Since the options are linearly independent, Bob can unambiguously distinguish them.\end{proof}

\section{Hilbert Module Structures from LOCC}\label{S:hilbertmodule}

In this section we note how an important special case in our analysis generates Hilbert $C^*$-module structures. Consider the case in one-way LOCC for which the number of states to be distinguished is equal to the dimension $d$ of the qudit base system; in other words $W$ is a unitary map or equivalently $r=d$ in the notation of Proposition~\ref{prop1}. Through a unitary change of basis induced by $W$ acting on the basis that defines the operation $\Delta$, we may assume $W=I$. Let $\mathfrak{A}$ be the $C^*$-algebra of operators on $\C^d$ with diagonal matrix representations in the basis that defines $\Delta$. Then $\mathfrak{S}$ given as above as the span of the operators $U_j^*U_i$ is a right (or left) $\mathfrak{A}$-module. We can define a map $\langle \cdot , \cdot \rangle : \mathfrak{S} \times \mathfrak{S} \rightarrow \mathfrak{A}$ by $\langle X , Y \rangle = \Delta(Y^* X)$, and one can check that this endows $\mathfrak{S}$ with the structure of a Hilbert $\mathfrak{A}$-module, for which the unitary operators $U_i$ form an orthogonal basis in this $\mathfrak{A}$-valued inner product.


 Given elements $U,V \in \mathfrak{S}$. We say that $U \sim V$ if there exists an invertible diagonal matrix $D$ such that $U = VD$. This is immediately seen to be an equivalence relation, one which respects orthogonality.

 \begin{lem}\label{LemmaEquivalence}
 Using the inner product defined above, the orthogonal complement is invariant under this equivalence relation: If $U \sim V$, then $\langle U,X \rangle = 0$ if and only if $\langle V,X \rangle = 0$ \end{lem}
\begin{proof}
Consider the following calculation:
\be \langle U,X \rangle = \Delta(X^*U) =  \Delta(X^*VD) =  \Delta(X^*V)D = \langle V,X \rangle D , \ee
from which the lemma follows.
\end{proof}

This means that any  set which is orthogonal with respect to $\langle \cdot , \cdot \rangle$ can contain at most one representative of each equivalence class, putting a significant bound on the size of any orthogonal set. In particular, we immediately recover the standard bound on the number of LOCC-distinguishable maximally entangled states.

\begin{prop}\label{unitaries}
Let $\{U_k\}_{k = 1}^n$ be a set of $d \times d$ unitary matrices which are mutually orthogonal under the inner product $\langle \cdot , \cdot \rangle$. Then $n \le d$.
\end{prop}
\begin{proof}
If $\langle U_i,U_j \rangle =0$, then $\tr (U_j^*U_i) = \tr (\Delta(U_j^*U_i)) = 0$, and so orthogonality under $\langle \cdot , \cdot \rangle$ implies orthogonality in the Hilbert-Schmidt inner product. Let $\{D_1, D_2, \ldots, D_{d}\}$ be a set of diagonal unitary matrices with $\tr (D_j^*D_i) = d\,\delta_{i,j}$. We look at the set ${\mathcal B} = \{U_kD_i\}$ for $k \in [n]$ and $i \in [d]$. From the lemma, $\langle U_kD_i,U_lD_j \rangle = \langle U_k,U_l\rangle = 0$ if $k \neq l$, which means that $\tr ((U_lD_j)^*U_kD_i) = 0$.

On the other hand, $\tr ((U_kD_j)^*U_kD_i) = \tr (D_j^*D_i) =d \,\delta_{i,j}$. This implies that ${\mathcal B}$ is a set of $nd$ matrices which are orthogonal in the Hilbert-Schmidt inner product, giving us $\vert {\mathcal B} \vert = nd \le d^2$ and $n \le d$.
\end{proof}

We can generalize this to a stronger statement which appears to be not as well known. Here we define $\langle X , Y \rangle_W= \Delta(W^*Y^* XW)$ for any  $r \times d$ matrix $W$ with $r \ge d$. We say that a set of vectors in $\mathbb{C}^d$ is generic if any $d$ of them are linearly independent.
\begin{prop}\label{prop: rank}
Let $\{M_k\}_{k = 1}^n$ be a set of $d \times d$  matrices which are mutually orthogonal under the inner product $\langle \cdot , \cdot \rangle_W$. If for each $M_k$, we have $\mbox{rank}(M_k) = r_k$, then $\sum_k r_k \le d^2$ as long as the columns of $W$ are generic.
\end{prop}
\begin{proof}
We first note that since each $r_k \le d$, we have $\sum_k r_k \le dn \le d^2$ as long as $n \le d$. So we need only consider the case $n>d$. As before, orthogonality under $\langle \cdot , \cdot \rangle_W$ implies that the set $\{M_kW\}$ is mutually orthogonal in the Hilbert-Schmidt inner product.
For each $k$, if $(W^*M_k^*M_kW)_{ii} = 0$ then $W_i \in \ker(M_k)$. Since the columns of $W$ are generic, the number of zeroes is bounded by the dimension of the kernel of $M_k$. Hence, the matrix $\Delta(W^*M_k^*M_kW)$ has at least $r - (d-r_k)$ nonzero entries. We can then find a set of $(r-d+r_k)$ diagonal matrices $D_{k,i}$ with support on the nonzero entries of $\Delta(W^*M_k^*M_kW)$ such that $\tr (D_{kj}^*W^*M_k^*M_kWD_{ki}) = \delta_{i,j} \tr W^*M_k^*M_kW$; hence these are orthogonal to each other and must be linearly independent.

This implies that ${\mathcal B} = \{M_kWD_{k,j}\}$ is a set of $n(r-d) + \sum_k r_k$ matrices which are orthogonal in the Hilbert-Schmidt norm on $r \times d$ matrices, giving us $\vert {\mathcal B} \vert = n(r-d) + \sum_k r_k \le rd$ and \bee \sum_k r_k \le (d-r)n + rd=   d^2 + (r-d)(d-n) \le d^2\eee since $n > d$ and $r \ge d$, and this completes the proof.
\end{proof}
\section{Conclusion}

This investigation was initially motivated by an attempt to better understand the mathematical foundation that underlies the operator relations  characterizing LOCC, and in particular the one-way classical communication version of the subject. Somewhat to our surprise, we found a variety of operator structures present in the background; specifically, operator algebras, operator systems, and Hilbert $C^*$-modules. Adding to this perspective some tools from matrix theory, we were able to establish some new results on perfect distinguishability of quantum states under different communication schemes, and we discovered new derivations for some established results and dimension bounds in the field.

The identification of operator structures and the introduction of operator and matrix theoretic approaches has paid dividends over the past decade and a half in a growing number of areas in quantum information; including now quantum error correction, quantum privacy, and the study of Bell inequalities via operator space techniques to name a few. We view this work as the initiation of a potentially fruitful line of investigation motivated by this operator and matrix theoretic perspective, and as an invitation for other researchers to participate and explore.

\vspace{0.1in}

{\noindent}{\it Acknowledgements.} D.W.K. was partly supported by NSERC and a University Research Chair at Guelph. C.M. was partly supported by a University of Guelph scholarship and the African Institute for Mathematical Sciences. R.P. was partly supported by NSERC. M.N. acknowledges the ongoing support of the Saint Mary's College Office of Faculty Research. The authors are grateful for the organizers of ILAS2016 in Leuven, Belgium, where this collaboration was initiated.

\end{document}